\documentclass{article}
\usepackage{amsmath}
\usepackage{xcolor}
\usepackage{amssymb,amsfonts}
\usepackage{graphicx}
\usepackage{enumitem}
\usepackage{mathrsfs}
\usepackage{hyperref}[colorlinks=true,urlcolor=blue,citecolor=magenta,citebordercolor=blue,linkcolor=blue]
\usepackage[a4paper,left=3cm,right=3cm]{geometry}

\newtheorem{theorem}{Theorem}[section]

\newtheorem{assumption}[theorem]{Assumption}

\newtheorem{corollary}[theorem]{Corollary}

\newtheorem{definition}[theorem]{Definition}

\newtheorem{lemma}[theorem]{Lemma}

\newtheorem{proposition}[theorem]{Proposition}
\newtheorem{remark}[theorem]{Remark}

\newtheorem{setting}[theorem]{Setting}

\newenvironment{proof}[1][Proof]{\textbf{#1.} }{\ \rule{0.5em}{0.5em}}

\newcommand{\R}{\mathbb{R}}

\newcommand{\EE}{\mathbb{E}}

\newcommand{\mcC}{\mathcal{C}}

\newcommand{\mcF}{\mathcal{F}}

\newcommand{\bfY}{\mathbf{Y}}
\newcommand{\bfX}{\mathbf{X}}

\newcommand{\bfx}{\mathbf{x}}
\newcommand{\bfz}{\mathbf{z}}
\newcommand{\bfZ}{\mathbf{Z}}
\newcommand{\bfw}{\mathbf{w}}
\newcommand{\bfW}{\mathbf{W}}
\newcommand{\bfy}{\mathbf{y}}

\newcommand{\bfu}{\mathbf{u}}

\newcommand{\bfe}{\mathbf{e}}

\newcommand{\bfone}{\mathbf{1}}
\newcommand{\bfzero}{\mathbf{0}}
\newcommand{\bflambda}{\mathbf{\lambda}}
\newcommand{\Af}{\square^{A}\hspace{-0.08cm}f }

\newcommand\ale[1]{{#1}}

\newcommand{\UU}{\mathbb{U}}

\newcommand{\probp}{{{P}}}

\newcommand{\norm}[1]{\left\lVert#1\right\rVert}
\newcommand\abs[1]{\left|#1\right|}
\newcommand\Ep[1]{\mathbb{E} \left[#1\right]}

\newcommand{\Linfty}{L^\infty}

\newcommand{\spU}{\square_{\pi}U}
\renewcommand\emph[1]{\textit{#1}}
\begin{document}

\author{Alessandro Doldi\thanks{Department of Mathematics, Università degli Studi di Milano, alessandro.doldi@unimi.it} ,  Marco Frittelli\thanks{Department of Mathematics, Università degli Studi di Milano, marco.frittelli@unimi.it}  and Emanuela Rosazza Gianin\thanks{Department of Statistics and Quantitative Methods, Università degli Studi di Milano-Bicocca, emanuela.rosazza1@unimib.it.
{ All authors are members of GNAMPA-INDAM}}}
\title{Are Shortfall Systemic Risk Measures One Dimensional?}

\maketitle

\abstract{
\noindent Shortfall systemic (multivariate) risk measures $\rho$ defined through an $N$-dimensional multivariate utility function $U$ \ale{and random allocations} can be represented as  classical (one dimensional) shortfall risk measure\ale{s} associated to an explicitly determined $1$-dimensional function constructed from $U$. This finding allows \ale{for simplifying}  the study of several properties of $\rho$, such as dual representations,  law invariance and stability.
}
\medskip



\section{Introduction}
We consider risky financial positions $(X^1,\dots, X^N):=\bfX$ and assume, to simplify the exposition in this introduction, that  $\bfX\in (L^{\infty}(\Omega, \mathcal F, P))^N:=(\Linfty)^N$.  We also let $\pi: (\Linfty)^N \to \Linfty$ be a pricing functional, $U:\R^N\rightarrow \R$ be a multivariate utility function and we set $\UU(\bfX):=\Ep{U(X)}$ and $\mathcal{C} := \{\bfY \in (\Linfty)^N \mid \pi(\bfY) \in \mathbb{R} \}$. Then the functional $\rho_{\pi, \UU}(\bfX) :(\Linfty)^N \to [-\infty,+\infty]$ { defined by}
\begin{equation}
\label{rhopu}
\rho_{\pi, \UU}(\bfX) := \inf \left\{ \pi(\bfY)\mid  \bfY \in \mathcal C, \, \UU(\bfX+\bfY) \geq 0  \right\}, \quad \bfX \in (\Linfty)^N
\end{equation}
represents a general capital requirement, as introduced in \cite{FrittelliScandolo06}, as well as a shortfall systemic risk measure, as extensively analysed in \cite{Drapeau}, \cite{BFFMB19} and \cite{BFFMB20}. A related, but alternative approach, based on set-valued maps, is considered in  \cite{FeinsteinRudloffWeber}. Observe that the amount $\pi(\bfY)$ is enforced to be deterministic, even though the terminal-time allocations $\bfY \in (\Linfty)^N$ are allowed to be scenario-dependent.
The map $\pi(\bfY)=\sum_{i=1}^N Y^i$ is a classical example for a pricing functional. Frequently used multivariate utility functions have the form $U(\bfx)=\sum_{i=1}^N U^i(x^i)$ for univariate utility functions $U^i:\R \to \R$, as in e.g. \cite{BFFMB20}, where also a detailed discussion on scenario-dependent allocation can be found.  A conditional version of \eqref{rhopu}  was treated in \cite{DoldiFrittelli21}. 

In this paper, we aim at establishing whether  the functional $\rho_{\pi,\UU}$ can be reduced to a classical \emph{univariate} shortfall risk measure
\begin{equation}
\rho_{\Ep{g}}(X):= \inf \left\{ \alpha \in \R \mid  \Ep{g(X+ \alpha)} \geq 0  \right\}, \quad X \in L^{\infty},
\end{equation}
based on some function $g:\R \to \R$ that can be \emph{explicitly recovered} from $U$ and $\pi$.
As a corollary of one of our main results, Theorem \ref{thmfull}, we show that under suitable conditions on $\pi$ and $U$
\begin{equation} \label{rhopuisunivariate}
\rho_{\pi, \UU}(\bfX)=\rho_{\Ep{g}}(\pi(\bfX))={\inf \left\{ \alpha \in \R \mid  \Ep{g(\pi(\bfX))+ \alpha)} \geq 0  \right\} }, \quad \bfX \in (\Linfty)^N,
\end{equation}
for $g(x)=\square_{\pi} U(x)
 := \sup \{U(\bfw)\mid \bfw \in \R^N, \pi(\bfw)=x \}, \, x\in \R$.
Moreover, we prove that $\rho_{\pi, \UU}(\bfX)$ admits a (unique) optimum $\bfY_\bfX$, given by   $\bfY_\bfX=-\bfX+\Theta(\pi(\bfX)+\rho_{\pi, \UU}(\bfX)))$ for a continuous explicit function $\Theta: \R \to \R^N$ depending on $U$ and $\pi$.

 The fact that a particular class of shortfall systemic risk measures (based on the particular choices  $U(\bfx)=\sum_{i=1}^N U^i(x^i)$ and $\pi(\bfY)=\sum_{i=1}^N Y^i$) could be reduced to univariate ones, was already observed in \cite{BFFMB19} Proposition 5.3, \cite{BFFMB20} Proposition 3.1 (ii),  as well as in \cite{OverbeclSchindler21} Theorem 3.16.
 One main difference from our work is that we provide very explicit representations (i.e. \eqref{rhopuisunivariate} and the more general case in Theorem \ref{thmfull} below) and produce explicit formulas for the function $g$ and for the optimum.

The representation \eqref{rhopuisunivariate} is obtained as a particular instance of a more general result. Indeed,  in \eqref{rhopu} $\UU$ needs not be an expected utility, but can rather be taken to be a general (multivariate) utility functional $\UU:(\Linfty)^N\rightarrow \R$. In such a case, we prove that \eqref{rhopuisunivariate} takes the form 
\begin{equation}
    \label{functionalrho}
    \rho_{\pi, \UU}(\bfX)= \inf \left\{ \alpha \in \R \mid  \square_{\pi}\UU(\pi(\bfX)+\alpha) \geq 0  \right\},\quad \bfX \in (\Linfty)^N,
\end{equation}
where $\square_{\pi} \UU (Z) := \sup \{\UU(\bfW): \bfW \in (\Linfty)^N, \pi(\bfW)=Z \}$, $Z \in \Linfty $, is the functional counterpart of $\square_{\pi}U$. Our results cover also cases when (i) we allow for unbounded positions $\bfX$ and (ii) $\pi$ is $\R^M$-valued, for some $M\geq 1$. When $M\geq 2$, the analogue of \eqref{rhopuisunivariate}, namely \eqref{neweq1} or its {functional counterpart  \eqref{333} } in the more general case, can be seeing as a dimensionality reduction property induced by random allocations. In this case, the RHS of \eqref{rhopuisunivariate} takes the form of a shortfall type systemic risk measure with deterministic allocations, i.e. in the form \eqref{eqrhog} below, as those treated e.g. in \cite{Drapeau}.  The  case $M\geq 2$ also covers grouping examples, in which terminal-time exchanges are allowed only within certain subgropus of the whole system (see Section \ref{secgroups}).

One first application of these findings is the law invariance of  multivariate shortfall risk measures in the form \eqref{rhopu} (see Section \ref{seclawinvariance}).
 Section \ref{secstability} is devoted to establishing  a Law of Large Numbers - type result in the style of \cite{Shapiro13} for systemic shortfall risk measures in the form \eqref{rhopu}.   Our approach here is inspired by the one of \cite{BurzoniDoldiCompagnoni23}.

\section{Systemic risk measures can be reduced to univariate risk measures } \label{sec: dimens reduct}

We consider a vector subspace $L$ of $L^0(\Omega, \mathcal F, P)$, with $\R\subseteq L$, and we induce on the Cartesian product $L^N$, $N \geq 1$, the order from the standard componentwise $P$-a.s. ordering from $(L^0(\Omega, \mathcal F, P))^N$.  For $N \geq M \geq 1$, 
let $\pi=(\pi^1,\dots,\pi^M)^T: L^N \to L^M$ and set $\mathcal{C} := \{\bfY \in L^N \mid \pi(\bfY) \in \mathbb{R}^M \}$.

Given the functions $\UU:L^N \to [-\infty, + \infty)$ and $G:L^M \to [-\infty, + \infty)$ we define $\rho_{\pi, \UU}:L^N \to [-\infty, + \infty]$ and $\rho_{G}:L^M \to [-\infty, + \infty]$ by

\begin{align}
\rho_{\pi, \UU}(\bfX) &:= \inf \left\{ \sum_{m=1}^M\pi^m(\bfY)\mid \bfY \in \mathcal{C},\UU(\bfX+\bfY) \geq 0  \right\}, \quad \bfX \in L^N, \notag \\
\rho_{G}(\bfX)&:= \inf \left\{ \sum_{m=1}^M\alpha^m\mid \alpha\in\R^M ,\, G(\bfX+ \alpha) \geq 0  \right\}, \quad \bfX \in L^M. \label{eqrhog}
\end{align}

Let $\bfe^i$ be the $i$-th element of the canonical basis of $\R^N$. 
We say that $\UU:L^N \to [-\infty, + \infty)$  is strictly increasing in some component, if for any $\bfX \in L^N$ there exists some $j \in \{1,\dots,N\} $ such that the function $x \to \UU(\bfX+\mathbf{e}^jx)$, $x \in \R$, is strictly increasing. 
\begin{assumption} \label{mass: F, pi, C}.\\
\noindent a) $\pi: L^N \to L^M$ is linear and satisfies, for $L_+=L\cap L_+^0(\Omega, \mathcal F, P)$
\begin{equation}\label{mpi1}
    \pi\left(L_+^N \right)=L_+^M \quad \mbox{ and } \quad \pi(\R_+^N)=\R_+^M.
\end{equation}

\noindent b) $\UU:L^N \to \mathbb{R}$  is concave, increasing and strictly increasing in some component. 
\end{assumption}

\begin{remark}
\label{imageofRNisRM}

Observe that \eqref{mpi1} implies that $\pi$ is monotone and that
    \begin{equation}\label{mpi2}
    \pi\left(L^N \right)=L^M \quad \mbox{ and } \quad \pi(\R^N)=\R^M.
\end{equation}
Indeed, $\bfZ=\bfZ^+-\bfZ^-\in L^M$, for $\bfZ^{\pm}$ the componentwise positive and negative parts, so that $\bfZ^\pm=\pi(\bfX_\pm)$ for $\bfX_\pm\in L_+^N$ (by \eqref{mpi1}) and by linearity $\bfZ=\pi(\bfX_+-\bfX_-)$. The same works with deterministic vectors in particular, yielding the second equality.

\end{remark}
\ale{In the following we adopt the conventions $\inf\emptyset=+\infty,\sup\emptyset=-\infty$.}
\begin{definition}
 We call the function $\square_{\pi} \UU : L^M \to [-\infty,+\infty]$ defined by
\begin{equation*}
\square_{\pi} \UU (\bfZ) := \sup \{\UU(\bfW)\mid \bfW \in L^N, \pi(\bfW)=\bfZ \}, \quad \bfZ \in L^M
\end{equation*}  
the \emph{sup-convolution of $\UU$ under $\pi$}.
\end{definition}
In Section \ref{secproofprop} the analogous concept, defined for \emph{convex} functions $f$, is there named \emph{image function of $f$ under $\pi$}, a terminology mutuated from \cite{Ro70}. Our choice is motivated by the following observation. Take $M=1$, suppose that $\UU^i:L \to [-\infty,\infty)$, $i=1,\dots,N, $ are $N$ given univariate utility functions, consider the multivariate utility $\UU:L^N \to [-\infty,\infty)$ defined by $\UU(\bfW)=\sum_{i=1}^N\UU^i(W^i) $ and the functional $\pi:L^N \to L$ given by the sum of the components, namely $\pi(\bfW)=\sum_{i=1}^N W^i$. Then by computing the sup-convolution $\UU^1 \square \dots \square \UU^N$ of the functions $\UU^i$ we get $$ (\UU^1 \square \dots \square \UU^N) (Z)=\square_{\pi} \UU(Z), \quad Z \in L. $$ 

Observe that in case $M=1$, $\rho_{(\square_{\pi}\UU)}$ is a \emph{ classical (univariate) risk measure}.
One first interesting finding is that \emph{any} systemic risk measure in the form $\rho_{\pi, \UU}(\bfX)$ can be written as a univariate risk measure $\rho_{(\square_{\pi}\UU)}(\pi(\bfX))$ associated to the sup-convolution $\square_{\pi}\UU$, namely to an explicitly determined univariate function. For $1 \leq M <N $ we analogously obtain a reduction in dimensionality, as explicitly described in the following proposition, whose proof is in Section \ref{secprop3}.
\begin{proposition} \label{prop3}
 Suppose that Assumption \ref{mass: F, pi, C} holds true and that $\square_{\pi} \UU (\bfZ) < + \infty$ for every $\bfZ \in L^M$. Then
 \begin{enumerate}
     \item \label{propItem1} The functional $\square_{\pi} \UU $ is finite valued, concave and increasing on $L^M$.
     \item \label{propItem2} If $\bfX \in L^N$ satisfies $\sup \{\UU(\bfX+\bfy) \mid \bfy \in \mathbb R ^N \}>0$, then

 \begin{equation}\label{333}
\rho_{\pi, \UU}(\bfX)=\rho_{(\square_{\pi}\UU)}(\pi(\bfX)).     
 \end{equation}
 \end{enumerate}
 \end{proposition}

\begin{remark}
\label{remsuponlinfty}    The assumption in Item \ref{propItem2} of Proposition \ref{prop3} is automatic if $L=\Linfty$ and $\sup \{\UU(\bfy) \mid \bfy \in \mathbb R ^N \}>0$, by monotonicity of $\UU$. 
\end{remark}

\begin{remark} \label{remrem}   
Recall the general notion of a capital requirement $\rho_{\pi, \mathcal{A}}: L^N \to [-\infty, + \infty]$ (see \cite{FrittelliScandolo06}) and of a monetary risk measure $\rho_{\mathcal{B}}: L \to [-\infty, + \infty]$ (see  \cite{FollmerSchied2}) :
\begin{align}\label{rhorho}
\rho_{\pi, \mathcal{A}}(\bfX) 
&:= \inf \left\{ \pi(\bfY) \mid \bfY \in \mathcal{C},\bfX+\bfY \in \mathcal{A} \right\}, \quad \bfX \in L^N,\\
\rho_{\mathcal B}(Z) 
&:= \inf \left\{\alpha \in \mathbb R \mid Z+\alpha \in \mathcal{B}\right\}, \quad Z \in L, \label{rho1}
\end{align}
for some acceptance sets $\mathcal{A} \subseteq L^N$ and $\mathcal{B} \subseteq L$. 
If $\pi$ is linear and for $\mathcal{C} := \{\bfY \in L^N \mid \pi(\bfY) \in \mathbb{R} \}$ we have for all $\bfX \in L^N$
\begin{eqnarray*}
\rho_{\pi, \mathcal{A}}(\bfX)&=& \inf \left\{\alpha \in \mathbb{R}\mid \pi(\bfY)= \alpha  ,\bfX+\bfY \in \mathcal{A} \right\} \\
&=& \inf \left\{\alpha \in \mathbb{R} \mid \exists \bfW \in \mathcal{A} \mbox{ s.t. }\pi(\bfW)= \pi(\bfX)+\alpha \right\} \\
&=& \inf \left\{\alpha \in \mathbb{R} \mid \pi(\bfX)+\alpha \in \pi(\mathcal{A}) \right\}=\rho_{\pi(\mathcal{A})}(\pi(\bfX)). 
\end{eqnarray*}
Hence, any capital requirement (or systemic multivariate risk measure) of dimension $N$ in the form \eqref{rhorho} with $\pi$ \emph{linear} can be reduced to a classical univariate risk measure in the form \eqref{rho1}.
\end{remark}

\textbf{In the remaining of this section we work in the following}

\begin{setting}
\label{settingbase}
$\,$

\begin{enumerate}
\item We select $L=L^{\infty}$ 
\item \label{pequalsax}The linear functional $\pi: (L^{\infty})^N \to (L^{\infty})^M $ is assigned by $$\pi(\mathbf{X})=A\mathbf{X},$$ where the \ale{(deterministic)} matrix $A$ in $\R^{M\times N}$ satisfies $A(\R^N_+)=\R^M_+$.
\item The multivariate utility function $U:\R^N\rightarrow \R$ is nondecreasing (w.r.t. the componentwise order), differentiable, strictly concave throughout all $\R^N$ with $\sup \{U(\bfy) \mid \bfy \in \mathbb R ^N \}>0$.
\item \label{settingItem4}The functional $\UU:(L^{\infty})^N \to \mathbb{R}$ has the form 
\begin{equation*}
\UU(\bfX):=\EE[U(\bfX)], \text{   } \bfX \in (L^{\infty})^N.
\end{equation*}
\end{enumerate}
\end{setting}
We point out that $A(\R^N_+)=\R^M_+$ implies that $A$ has full rank, $\mathrm{rank}(A)=M$, and that $\pi$ in Item \ref{pequalsax} satisfies $\pi((\Linfty)^N_+)=(\Linfty)^M_+$. Moreover, the function $\UU$ in Item \ref{settingItem4} is also strictly increasing (in any component). Thus, in the Setting \ref{settingbase} the Assumption \ref{mass: F, pi, C} holds true.
The choices made in the Setting \ref{settingbase} lead to the classical shortfall systemic risk measure:
\begin{equation}\label{eqA}
\rho_{\pi, \UU}(\bfX) := \inf \left\{  \sum_{m=1}^M\pi^m(\bfY)\mid \bfY \in \mathcal{C},\Ep{U(\bfX+\bfY)} \geq 0  \right\}, \quad \bfX \in (L^{\infty})^N,
\end{equation}
which is a monotone increasing, convex, cash additive map. 
By definition, the sup-convolution $\square_{\pi} U: \R^M \to \R $  of $U$ under $\pi$ is assigned by:  
\begin{equation}\label{111}
\square_{\pi}U (\bfy) := \sup \{U(\bfx)\mid \bfx \in \R^N, \pi(\bfx)=\bfy \}, \quad \bfy \in \R^M.
\end{equation} 

\begin{assumption}
\label{assexput}
  For some $\bfy\in\R^M$,  the problem in \eqref{111}  admits an optimum, namely there exists $\bfx=\bfx(\bfy)\in \R^N$ such that $\pi(\bfx)=\bfy$ and $\square_{\pi} U(\bfy)=U(\bfx)$.  
\end{assumption}
We provide in Lemma \ref{lemmaexistsoptu} in Appendix mild conditions which guarantee the validity of Assumption \ref{assexput}.

The main result of this note is described in the following Theorem. Shortfall systemic risk measures $\rho_{\pi, \UU}(\bfX)$ defined through a $N$-dimensional multivariate utility function $U$ can be represented as a shortfall risk measure $\rho_{(\Ep{\square_{\pi} U})}(\pi(\bfX)) $ associated to the $M$-dimensional function $\square_{\pi} U$. Additionally, we provide the explicit formula for the optimum. The proof is deferred to Section \ref{sec:proffmainthm}.

\begin{theorem}
\label{thmfull}
Suppose that Assumption \ref{assexput} is satisfied. Then
\begin{enumerate}
    \item \label{ThItem1}

   $ \square_{\pi} \UU(\bfZ)=\Ep{(\square_{\pi}U)(\bfZ)}$  for every $\bfZ \in (L^{\infty})^M$.
\item \label{ThItem2} For every $\bfX \in (\Linfty)^N$ we have 
 \begin{equation}
\rho_{\pi, \UU}(\bfX)=\rho_{(\Ep{\square_{\pi} U})}(\pi(\bfX)):= \inf \left\{ \sum_{m=1}^M\alpha^m \mid \alpha \in \R^M, \, \Ep{\square_{\pi}U(\pi(\bfX)+\alpha)} \geq 0  \right\}.\label{neweq1}
 \end{equation}
If additionally there exists an optimum $\widehat{\alpha}\in \R^M$ attaining  the infimum in RHS of \eqref{neweq1},  then also $\rho_{\pi, \UU}(\bfX)$ admits a (unique) optimum $\bfY_\bfX$, given by
    \begin{equation}
        \label{formula:opt1U}
        \bfY_\bfX=-\bfX-\nabla U^*\Big(\pi^T \cdot \nabla (\square_{\pi}U)(\pi(\bfX)+\widehat{\alpha})\Big)=-\bfX+\Theta\Big(\pi(\bfX)+\widehat{\alpha}\Big),
    \end{equation} 
where \ale{$U^*$ is the concave conjugate of $U$, $\pi^T $ is the transposed map of $\pi$ and} $\Theta:\R^M \rightarrow \R^N $, $\Theta(\bfy):=-\nabla U^*\Big(\pi^T \cdot \nabla (\square_{\pi}U)(\bfy)\Big)$, is continuous.
\end{enumerate}
In case $M=1$,  
$ \rho_{\pi, \UU} $ is finite valued, the optimum in the RHS of \eqref{neweq1} always exists, and it is given by $\widehat{\alpha}=\rho_{(\Ep{\square_{\pi} U})}(\pi(\bfX))$.
\end{theorem}

{In the case $M=1$, as an immediate byproduct of Theorem \ref{thmfull}, 
the \emph{dual representation} for the systemic risk measure $\rho_{\pi,\UU}(=\rho_{(\Ep{\square_{\pi} U})})$ can be directly obtained from the well known classical dual representation of the univariate convex risk measure $\rho_{(\Ep{\square_{\pi} U})}$. Indeed, letting $\ell(x)=-\square _{\pi}U(-x),x\in\R$ the dual representation of $\rho_{(\Ep{\square_{\pi} U})}$ follows from \cite{FollmerSchied2} Theorem 4.115 with minimal penalty function $\alpha^\mathrm{min}$ in \cite{FollmerSchied2} Theorem 4.115 explicitly given, since $\ell^*(z)=-U^*(A^Tz), z\in\R$. }

 \subsection{Grouping case}
 \label{secgroups}
 
 In Lemma \ref{lemmagroups} below, whose simple proof is omitted, we show how the dimensionality reduction put in evidence in \eqref{neweq1} covers also the grouping case in Example 5.2 of \cite{BFFMB20} and Definition 5.1 of \cite{DoldiFrittelli21}. More precisely, we show that with an appropriate choice of $\pi$ we get:
 \begin{equation} 
  \label{eqgroups}
\mathcal{C} := \{\bfY \in L^N \mid \pi(\bfY) \in \mathbb{R}^M \}=\left\{\bfY\in(\Linfty)^N\mid \sum_{n\in I_m}Y^n\in\R\, \quad \forall\,m=1,\dots,M\right\}. 
  \end{equation}
 
\begin{lemma}
\label{lemmagroups}
  Let $I_1,\dots,I_M\subseteq \{1,\dots, N\}$ be a partition of $\{1,\dots,N\}$, clearly with $M\leq N$. Define the matrix $A=(a_{mn})_{mn}\in\R^{M\times N}$ via
  $$
  a_{mn}= 
  \begin{cases}
1&\text{ if } n\in I_m\\
0&\text{ otherwise}
  \end{cases}\quad\quad n=1,\dots, N;m=1,\dots, M.
  $$  
  Furthermore, set $L=\Linfty$ and $\pi(\bfX)=A\bfX$ (as a matrix-vector product). Then $A$ is full rank, the first Item in Assumption \ref{mass: F, pi, C} and the second Item in Setting \ref{settingbase} are satisfied, and  \eqref{eqgroups} holds.

\end{lemma}

\section{Applications}

\subsection{Law invariance}\label{seclawinvariance}

We use the same notation of Section \ref{sec: dimens reduct}, we write $P_X$ (resp. $P_\bfX$) for the law of a random variable $X$ (resp. vector $\bfX$) on $\R$ (resp. $\R^N$) and $X\overset{P}{\sim}Y$ if the random variables (or vectors) $X,Y$ have the same law under $P$.

\begin{proposition}\label{propLI}
Assume that $\pi(\bfX)\overset{P}{\sim}\pi(\bfY)$ with $\bfX,\bfY \in L^{N}$. Then
\begin{enumerate}
    \item[(1)] If $\pi(\mathcal{A})$ is law invariant then $\rho_{\pi, \mathcal{A}}(\bfX)=\rho_{\pi, \mathcal{A}}(\bfY)$.
    \item[(2)] If \eqref{333} holds and if $\mathcal{B}=\{\bfZ \in L^M \mid \square_{\pi} \UU (\bfZ) \geq 0\}$ is law invariant then $\rho_{\pi, \UU}(\bfX)=\rho_{\pi, \UU}(\bfY)$.
\end{enumerate}

\end{proposition}
\begin{proof}
Item (1) is an immediate consequence of $\rho_{\pi, \mathcal{A}}(\bfX)=\rho_{\pi(\mathcal{A})}(\pi(\bfX))$ proven in Remark \ref{remrem}. 
(2) From $\pi(\bfX)\overset{P}{\sim}\pi(\bfY)$ and the law invariance of $\mathcal{B}$ we get: $\square_{\pi} \UU (\pi(\bfX) + \alpha) \geq 0$ iff $\square_{\pi} \UU (\pi(\bfY) + \alpha) \geq 0$.
From $\rho_{\pi, \UU}(\bfX)=\rho_{(\square_{\pi}\UU)}(\pi(\bfX))$, we deduce that $\rho_{\pi, \UU}(\bfX)=\rho_{\pi, \UU}(\bfY)$.
\end{proof}
\begin{remark}  
Obviously, if $\pi$ is law invariant and $\bfX\overset{P}{\sim}\bfY$ then the assumption in the previous proposition holds, so that Proposition \ref{propLI} gives in particular sufficient conditions for the law invariance of the systemic risk measures $\rho_{\pi, \mathcal{A}}$ and $\rho_{\pi, \UU}$. \\ 
\end{remark}
\begin{corollary}
\label{corlawinv}
Let $U:\R^N\rightarrow \R$ be nondecreasing, differentiable, strictly concave throughout all $\R^N$ and satisfying $\sup \{U(\bfy) \mid \bfy \in \mathbb R ^N \}>0$ and Assumption \ref{assexput}. Then $\rho$ defined by
\begin{equation}
\label{eqsums}
\rho(\bfX) := \inf \left\{ \sum_{i=1}^N Y^i \mid \bfY \in (L^{\infty})^N, \, \sum_{i=1}^N Y^i \in \mathbb R, \, \Ep{U(\bfX+\bfY)} \geq 0  \right\}, \quad \bfX \in (L^{\infty})^N,
\end{equation}  
is finite valued and law invariant.
\end{corollary}
\begin{proof}
Observe that $\rho=\rho_{\pi, \UU}(\bfX) $ for  $\UU(\cdot):=\EE[U(\cdot)]$ and $\pi(\bfY):= Y_1+...+Y_N$, the latter being law invariant. Thus the assumptions in the Setting \ref{settingbase} hold. By Proposition \ref{prop:explicit formula}, $\, \square_\pi U$ is continuous on $\mathbb R$ and, by Theorem \ref{thmfull}, Item \ref{ThItem1},  $\square_{\pi} \UU (Z)= \Ep{ \square_\pi U(Z)}$. \ale{Thus,  $\mathcal{B}:=\{Z \in L^{\infty} \mid \square_{\pi} \UU (Z) \geq 0\}= \{ Z \in L^{\infty} \mid \Ep{\square_\pi U(Z)}\geq 0\}$}  is law invariant. The conclusion follows from Theorem \ref{thmfull} and Proposition \ref{propLI} (2).
\end{proof}

\subsection{Stability}
\label{secstability}
Let $(\Omega,\mcF,\probp)$ be an atomless standard probability space. If $\rho:(\Linfty)^N\rightarrow \R$ is a law invariant functional, then it is possible to think of $\rho$ as defined on the class of probability measures on $\R^N$. Indeed, whenever  $P_\bfX=P_\bfY$ on $\R^N$ we have $\bfX\overset{P}{\sim}\bfY$ and $\rho(X)=\rho(Y)$, and since the underlying space is non atomic every probability measure on $\R^N$ is realized as the law under $P$ of some random vector $\bfZ$ defined on $\Omega$ by the Skorokhod Theorem. By a slight abuse of notation we write $\rho(Q)$ meaning $\rho(\bfZ)$ for every $\bfZ$ having law  $Q$ on $\R^N$.
\begin{corollary}
\label{corcontonQ}
In Setup \ref{settingbase} with $M=1$ and $A\bfx=\sum_{j=1}^Nx^j$ for every $\bfx\in\R^N$, suppose that the assumptions of Corollary \ref{corlawinv} are satisfied. Take probability measures $\{Q_n\}_n$ on $\R^N$, for $n=1,\dots,+\infty$, such that for some positive radius $r>0$ we have 
$Q_n(B_r)=1$ for every $1\leq n\leq +\infty$, $B_r$ being the ball of radius $r$ in $\R^N$. Suppose additionally that $Q_n$ converges to $Q_\infty$ in the weak sense for probability measures. Then $\rho$ defined in  \eqref{eqsums} satisfies
$$\lim_{n\rightarrow +\infty}\rho(Q_n)=\rho(Q_\infty).$$
\end{corollary}
\begin{proof}
By the Skorokhod Representation Theorem there exist $N$-dimensional random vectors $\bfZ_n, 1\leq n\leq +\infty$ on $(\Omega,\mcF,P)$ such that $Q_n$ is the law of $\bfZ_n$ under $P$ and $\bfZ_n\rightarrow \bfZ_\infty$ $P$-a.s. In particular then $\bfZ_n\in (\Linfty)^N$, $\pi(\bfZ_n)\rightarrow \pi(\bfZ_\infty)$ and $\norm{\pi(\bfZ_n)}_\infty\leq Nr\,P-$a.s. for every $n$.  Then for $1\leq n\leq +\infty$ we have $\rho(Q_n)=\rho(\bfZ_n)=\rho_{(\Ep{\square_{\pi} U})}(\pi(\bfZ_n))$ by Theorem \ref{thmfull}, where $\square_\pi U:\R\rightarrow \R$ is increasing and nonconstant by Proposition \ref{prop:explicit formula}. We now show that $\rho_{(\Ep{\square_{\pi} U})}(\pi(\bfZ_n))\rightarrow_n\rho_{(\Ep{\square_{\pi} U})}(\pi(\bfZ_\infty))$. By \cite{FollmerSchied2}  Proposition 4.113 $\rho_{(\Ep{\square_{\pi} U})}$ is continuous from below. Then it has the Lebesgue property (\cite{FollmerSchied2}  Corollary 4.35) and the desired convergence follows.
\end{proof}

Take now $\bfX\in (\Linfty)^N$. Replacing  the law $P_\bfX$ in $\rho(P_\bfX)$ with the empirical
measure $\widehat{P}_n$ based on an i.i.d. sample $(\bfX_1,\dots, \bfX_n)$, we obtain the empirical estimate\slash historical estimate 
 $\rho(\widehat{P}_n)$ of $\rho(P_\bfX)=\rho(\bfX)$. 
Under the assumptions of Corollary \ref{corcontonQ}, we have in particular $\lim_n\rho(\widehat{P}_n)=\rho(P_\bfX)$ $P-$a.s. by weak convergence (a.s.) of $\widehat{P}_n$ to $P_\bfX$. This can be exploited in conjunction with the explicit formula for the optima \eqref{formula:opt1U}, since in the case $M=1$ we know $\widehat{\alpha}=\rho(\bfX)$, to guarantee a.s. convergence of the approximated optimal allocation functions $\bfx\mapsto -\bfx+\Theta(\pi(\bfx)+\rho(\widehat{P}_n))$.

\section{Proof of Proposition \ref{prop3}} \label{secprop3}

\textbf{Proof of Item \ref{propItem1}}. 
The functional $\square_{\pi} \UU$ is finite valued since, by \eqref{mpi2}, $\square_{\pi} \UU(\bfZ) >- \infty$ for any $\bfZ \in L^M$.

\textit{Concavity.} By \eqref{mpi2}, given any $\bfZ_1, \bfZ_2 \in L^M$, there exist $\bfW_1, \bfW_2 \in L^N$ such that $\pi(\bfW_i)=\bfZ_i$ for $i=1,2$. 
Hence, for any $\alpha \in [0,1]$
\begin{equation*}
\square_{\pi} \UU (\alpha \bfZ_1 + (1- \alpha) \bfZ_2)  \geq \UU(\alpha \bfW_1 + (1-\alpha) \bfW_2) 
\geq \alpha \UU( \bfW_1) + (1-\alpha) \UU(\bfW_2), 
\end{equation*}
where the former inequality is due to the definition of $\square_{\pi} \UU$ and the linearity of $\pi$, the latter from concavity of $\UU$.
Concavity of $\square_{\pi} \UU$ then follows by taking the supremum over all $\bfW_1, \bfW_2 \in L^N$ such that $\pi(\bfW_1)=\bfZ_1$ and $\pi(\bfW_2)=\bfZ_2$.

\textit{Monotonicity.} 
Take $\bfZ_i \in L^M$ such that $\bfZ_1 \leq \bfZ_2$ and take by \eqref{mpi2} $\bfW_1\in L^N$ s.t. $\pi(\bfW_1)=\bfZ_1$.
Now $\bfZ_2-\bfZ_1\in L_+^M$, so that $\bfZ_2-\bfZ_1=\pi(\bfW)$ for some $\bfW\in (L_+)^N$ by \eqref{mpi1}. Hence, $\bfW_2:=\bfW_1+\bfW\geq \bfW_1$ satisfies $\pi(\bfW_2)=\bfZ_2$ and \ale{$  \UU(\bfW_1) \leq \UU(\bfW_2) \leq \square_{\pi}\UU(\bfZ_2).$}
Take now a supremum over $\bfW_1$ satisfying $\pi(\bfW_1)=\bfZ_1$ to get $\square_{\pi}\UU(\bfZ_1)\leq \square_{\pi}\UU(\bfZ_2)$.
\\
\textbf{Proof of Item \ref{propItem2}}. \ale{Observe first that under the additional assumption in Item \ref{propItem2}, we have $\rho_{\pi, \UU}(\bfX)<+\infty$.}
From the linearity of $\pi$ and the definition of $\mathcal C$, we have for any $\bfX \in L^N $
\begin{align}
\rho_{\pi, \UU}(\bfX)&=\inf \left\{ \sum_{m=1}^M\pi^m(\bfY)\mid \bfY \in \mathcal{C},\UU(\bfX+\bfY) \geq 0  \right\} \notag \\ \label{m111}
&= \inf \left\{ \sum_{m=1}^M\pi^m(\bfY)\mid \bfY \in \mathcal{C},\UU(\bfX+\bfY) > 0  \right\} \\ 
&= \inf \left\{ \sum_{m=1}^M\alpha^m\mid \alpha\in \R^M\ale{\text{ satisfies } \, \pi (\bfY)=\alpha   \text{ for some }\bfY\in L^N},\UU(\bfX+\bfY) > 0  \right\} \notag \\
&= \inf \left\{\sum_{m=1}^M\alpha^m\mid \alpha\in \R^M,\, \exists \bfW \in L^N, \UU(\bfW)>0 \mbox{ s.t. }\pi(\bfW)= \pi(\bfX)+\alpha \right\} \notag \\ \label{m222}
&= \inf \left\{ \sum_{m=1}^M\alpha^m\mid \alpha\in \R^M,\, \sup \left \{{\UU(\bfW) \mid \bfW \in L^N , \pi(\bfW)= \pi(\bfX)+\alpha } \right \}>0 \right\} \notag \\
&= \inf \left\{ \sum_{m=1}^M\alpha^m\mid \alpha\in \R^M,\, \square_{\pi}\UU( \pi(\bfX)+\alpha)>0  \right\} \notag \\ 
&= \inf \left\{ \sum_{m=1}^M\alpha^m\mid \alpha\in \R^M,\, \square_{\pi}\UU( \pi(\bfX)+\alpha) \geq 0  \right\}  \\
&=\rho_{(\square_{\pi}\UU)}(\pi(\bfX)) \notag 
\end{align}
where only the equalities \eqref{m111}  and \eqref{m222} are not evident. 
To prove the equality in \eqref{m111}, observe first that, obviously, $$\rho_{\pi, \UU}(\bfX) \leq \inf \left\{ \sum_{m=1}^M\pi^m(\bfY)\mid \bfY \in \mathcal{C},\UU(\bfX+\bfY) > 0  \right\}:=a$$ \ale{with $a\in\R$.}  Suppose by contradiction that $\rho_{\pi, \UU}(\bfX)<a$ and let $0<\varepsilon<\frac{a-\rho_{\pi, \UU}(\bfX)}{2
}  $. Then there exists $\bfY \in \mathcal{C}$ such that  $\UU(\bfX+\bfY) \geq 0$ and 
\begin{equation}
\label{mpropertY}
    \sum_{m=1}^M\pi^m(\bfY) < \rho_{\pi, \UU}(\bfX)+\varepsilon.
\end{equation}
  Since $\UU$ is strictly increasing on one component, say component $i$, take $\widehat{\bfY}:=\bfY+\varepsilon \frac{\bfe^i}{\sum_{m=1}^M\pi^m(\bfe^i)+1}$, noticing that this is well defined as $\pi^m(\bfe^i)\geq 0$, for all $m$, by \eqref{mpi1}. Since $\pi$ is linear and $\pi(\mathbb R^N) \subseteq \R^M$ (Remark \ref{imageofRNisRM}) then $\sum_{m=1}^M\pi^m(\widehat{\bfY})=\left (\sum_{m=1}^M\pi^m(\bfY) +\varepsilon \frac{\sum_{m=1}^M\pi^m(\bfe^i)}{\sum_{m=1}^M\pi^m(\bfe^i)+1} \right ) \in \R$ and $\widehat \bfY \in \mathcal{C}$. Moreover,  $\UU(\bfX+\widehat{\bfY})=\UU(\bfX+\bfY+\varepsilon \frac{\bfe^i}{\sum_{m=1}^M\pi^m(\bfe^i)+1})>\UU(\bfX+\bfY) \geq 0$ so that $a \leq \sum_{m=1}^M\pi^m(\widehat{\bfY})$. But this is a contradiction using \eqref{mpropertY}: $a \leq \sum_{m=1}^M\pi^m(\widehat{\bfY})=\sum_{m=1}^M\pi^m(\bfY) +\varepsilon \frac{\sum_{m=1}^M\pi^m(\bfe^i)}{\sum_{m=1}^M\pi^m(\bfe^i)+1}<\rho_{\pi, \UU}(\bfX)+2\varepsilon <a$.
To prove the equality in \eqref{m222}, we set
\begin{equation*}
g_{\bfX}(\alpha) := \square_{\pi} \UU(\pi(\bfX) + \alpha), \quad \alpha \in \R^M 
\end{equation*}
and show that
 \begin{equation*}
\inf \left\{\sum_{m=1}^M\alpha^m \mid\alpha \in \R^M,\, g_{\bfX} (\alpha)>0 \right\}=\inf \left\{\sum_{m=1}^M\alpha^m\mid \alpha\in  \R^M,\, g_{\bfX} (\alpha) \geq 0 \right\}.
\end{equation*}
The fact that LHS$\geq$RHS is clear, \ale{and the equality is trivial if the set in RHS is empty. Then we assume this is not the case and prove LHS$\leq$RHS.} Take a minimizing sequence $(\alpha_n)_n$ such that $g_{\bfX} (\alpha_n) \geq 0$ for each $n$ and $\sum_{m=1}^M\alpha_n^m\downarrow_n\inf \{\sum_{m=1}^M\alpha^m\mid \alpha\in  \R^M,\, g_{\bfX} (\alpha) \geq 0 \}$.

\textbf{Case 1:}
 $(\alpha_n)_n$ admits a subsequence $(\alpha_{n_k})_k$ with $g_{\bfX} (\alpha_{n_k}) > 0$ for each $k$. Clearly we get LHS$\leq \sum_{m=1}^M\alpha_{n_k}^m\downarrow_k$RHS, which is the desired remaining inequality. 
 
 \textbf{Case 2:} $g_{\bfX} (\alpha_n) = 0$ definitely in $n$.
  We assume the equality holds for each $n$ w.l.o.g.. 
Define now the functions $h_n(\beta):=g_{\bfX}(\alpha_n+\beta \bfone)$, $\beta \in \R$, and observe that Proposition \ref{prop3} Item \ref{propItem1} ensures that, for each $n$,  $h_n:\R\rightarrow \R$ is increasing and concave on $\R$ and thus continuous. Moreover, $$g_\bfX(\alpha)\leq h_n\left(\max_m\abs{\alpha^n_m}+\max_m\abs{\alpha_m}\right),\quad \forall\,\alpha\in\R^M$$
guaranteeing that $\sup_{\beta\in\R}h_n(\beta)=\sup_{\alpha\in\R^M}g_\bfX(\alpha):=\widehat g $. By the linearity of $\pi$ we then get 
\begin{align}
\sup_{\beta\in\R}h_n(\beta)&=\sup_{\alpha \in \mathbb{R}^M} g_{\bfX} (\alpha)
=  \sup_{\alpha \in \mathbb{R}^M }  \left \{\sup \{\UU(\bfW) \mid \bfW \in L^N, \pi(\bfW)=\pi (\bfX)+\alpha \} \right \} \notag \\
&=  \sup_{\alpha \in \mathbb{R}^M }  \left \{\sup \{\UU(\bfX+\bfY) \mid \bfY \in \mcC, \pi(\bfY)=\alpha \} \right \} \notag \\
&=    \sup \{\UU(\bfX+\bfY) \mid \bfY \in \mcC \} \label{ma1} \\
&\geq    \sup \{\UU(\bfX+\bfy) \mid \bfy \in \mathbb R ^N \}>0, \label{ma2}
\end{align}
where we used: 
in \eqref{ma1} the equality \ale{$ \pi(\mcC)= \mathbb R^M$} (a consequence of \eqref{mpi2}); in the first inequality in \eqref{ma2}  the fact that $\pi (\mathbb R ^N) \subseteq \mathbb R^M$, and the  last strict inequality is guaranteed \ale{by hypothesis}.
Observe that since $g_{\bfX} (\alpha_n) = 0$, we also have $h_n(0)=0<\sup_{\beta\in\R}h_n(\beta)=\widehat g$.
Let $\widehat \beta_n = \inf \{ \beta \in \R \mid h_n(\beta) = \widehat g \} \leq + \infty $. 
From $h_n(0)<\widehat g$, the continuity and monotonicity of $h_n$, we have $\hat{\beta_n}> 0$ for every $n$.
Additionally, $h_n$ as a univariate concave and increasing function is strictly increasing on $(-\infty, \widehat \beta_n)$. Thus, for some $0<\varepsilon_n < \min(\frac1n,\widehat\beta_n)$ we have $0=h_n(0)<h_n(\varepsilon_n)=g_{\bfX}(\alpha_n+\varepsilon_n \bfone)$. Thus $\beta_n:=\alpha_n+\varepsilon_n\bfone$ defines a minimizing sequence, with $g_{\bfX}(\beta_n)>0$:
$$\sum_{m=1}^M\beta_n^m=\sum_{m=1}^M\alpha_n^m+M\varepsilon_n\downarrow_n\inf \{\sum_{m=1}^M\alpha^m\mid \alpha\in  \R^M,\, g_{\bfX} (\alpha) \geq 0 \}$$ 
and one can argue as in Case 1.

\section{Image functions on $\R^N$} \label{secproofprop}

We now present some key properties of image functions. All definitions, as well as the notation, are mutuated from \cite{Ro70}. For convenience of the reader and to simplify the comparison with this reference, we thus opted to present the concepts and results for \emph{convex} functions $f$ and linear maps $A$, which replace the function $(-U)$ and the linear map $\pi$ in the previous sections.

\textbf{In this Section \ref{secproofprop} we work again in Setting \ref{settingbase} without further mention.}
We set $f=-U:\R^N\rightarrow \R$ and denote by $f^*$ the usual convex conjugate function of $f$, $$f^*(\bfz):=\sup_{\bfx\in\R^N}\left(\sum_{j=1}^Nx^jz^j-f(\bfx)\right)\in(-\infty,+\infty]\,,\quad \bfz\in\R^N.$$
\begin{definition}
The \emph{image function of $f$ under $A$} is the function $\Af:\R^M\rightarrow [-\infty,+\infty]$ 
\begin{equation}
    \label{imagefunction}
    \Af(\bfy):=\inf\left\{f(\bfx)\mid \bfx\in\R^N, A\bfx=\bfy\right\},  \text{   } \bfy\in\R^M,
\end{equation}
 with the usual convention $\inf\emptyset=+\infty$.
 We say that, for a given $\bfy\in\R^M$,  the problem in \eqref{imagefunction}  admits an optimum if there exists $\bfx=\bfx(\bfy)\in \R^N$ such that $A\bfx=\bfy$ and $\Af(\bfy)=f(\bfx)$. 
 \end{definition}
 In \cite{Ro70} the image function $\Af$ is denoted by $Af$. 
 Observe that $\Af=-\square_\pi U$, \ale{as in \eqref{111},} for $f=-U$ and $A=\pi$ and that
 $\Af $ is a convex function (Th. 5.7 \cite{Ro70}) and, since $f$ is real valued on the whole $\R^N$ and $A$ has full range, $\Af(\bfy)<+\infty$ for every $\bfy \in \R^M$. 
We stress that by strict convexity of $f$, the problem \eqref{imagefunction} admits at most one optimum. 
\begin{remark}
\label{exist:opt}
   Under Assumption  \ref{assexput}, we have that for some $\bfy\in\R^M$,  the problem in \eqref{imagefunction}  admits an optimum.
\end{remark}

We refer to \cite{Ro70} Chapter 26 for definitions of essentially smooth functions and Legendre type pairs. 
In our setting, $f$ is essentially smooth on the whole $\R^N$ and thus $(f,\R^N)$ is of Legendre type. We also briefly recall the following key results:
\begin{theorem}[\cite{Ro70} Theorem 26.5]
\label{Ro7026.5}
Let $h:\R^D\rightarrow(-\infty,+\infty]$ be a closed (i.e. \ale{proper} lower semicontinuous w.r.t. the usual Euclidean topology) convex function. Set $C:=\mathrm{int}\,\mathrm{dom}(h),C^*:=\mathrm{int}\,\mathrm{dom}(h^*)$. Then $(h,C)$ is a convex function of Legendre type iff so is $(h^*,C^*)$. When these conditions hold, $\nabla h$ is one-to-one from the open convex set $C$ onto $C^*$, continuous in both directions, and $(\nabla h)^{-1}=\nabla h^*$.  
\end{theorem}

 \begin{theorem}[\cite{Ro70} Theorem 16.3]
\label{Ro7016.3}
Suppose that there exists $\bflambda\in\R^M$ such that $A^T\lambda\in\mathrm{ri}\,\mathrm{dom}(f^*)$. Then for every $\bfy\in\R^M$ there exists an optimum for \eqref{imagefunction}.
\end{theorem}
\begin{proof}
Since $f$ is real valued and convex we have $f=(f^*)^*$ on $\R^N$. Moreover $(A^T)^T=A$ and thus we can rewrite 
\begin{align*}
    \inf\left\{f(\bfx)\mid \bfx\in\R^N, A\bfx=\bfy\right\}=\inf\left\{(f^*)^*(\bfx)\mid \bfx\in\R^N, (A^T)^T\bfx=\bfy\right\}
\end{align*}
Setting $g=f^*$, which is a convex function on $\R^N$, we recognize the setup of the last part in \cite{Ro70} Theorem 16.3 with $f,\lambda,\bfx,\bfy,A^T$ here in place of $g^*,\bfx,\bfy^*,\bfx^*,A$ in the reference respectively. By hypothesis we have $A^T\lambda\in\mathrm{ri}\,\mathrm{dom}(g)$ and so by \cite{Ro70} Theorem 16.3 the infimum in \eqref{imagefunction} is attained.
\end{proof}

\begin{proposition}
\label{prop:KuhnTucker}
    Take $\bfy\in\R^M$ such that $\Af(\bfy)>-\infty$. Then $\Af(\bfy)\in\R$ and the following are equivalent:
    \begin{enumerate}[label=(\roman*)]
        \item  there exists the optimum  $\bfx=\bfx(\bfy)\in \R^N$ for \eqref{imagefunction} 
        \item  there exist $({\bfx},\bflambda)=(\bfx(\bfy),\bflambda(\bfy))\in \R^N\times \R^M$  solving 
\begin{equation}
    \label{eq:KuhnTucker}
    \begin{cases}
         \nabla f({\bfx})=A^T\bflambda\\
         A\bfx=\bfy
        \end{cases}
\end{equation}
    \end{enumerate}
\end{proposition}
\begin{proof}
We already know that $\Af(\bfy)<+\infty$ for every $\bfy \in \R^M$, thus $\Af(\bfy)\in\R$. Observe that, once $\bfy\in\R^M$ is fixed, \eqref{imagefunction} is what is called in \cite{Ro70} Chapter 28 an ordinary convex program admitting a feasible solution (since $f$ is real valued). Its set of constraint is given by $\bfy-A\bfx=0$, 
and $C=\mathrm{ri}(C)=\R^N$ in the notation of \cite{Ro70}.  By the Kuhn-Tucker Theorem (\cite{Ro70} Corollary 28.3.1, whose hypotheses are met since we are assuming $\Af(\bfy)>-\infty$), (i) in the statement is equivalent to: there exists a pair $({\bfx},\bflambda)\in \R^N\times \R^N$ satisfying conditions (a),(b),(c) in \cite{Ro70} Theorem 28.3 (with $\bfx$ in place of $\overline{\bfx}$ and $\bflambda$ in place of $\bfu^*$). By the discussion following the proof of \cite{Ro70} Theorem 28.3, since $f$ and the functions enforcing the constraints  are differentiable, condition (c) can be rewritten as $\nabla f({\bfx})-A^T\bflambda=\bfzero$. Condition (b) is $A\bfx=\bfy$, and conditions (a) can actually be ignored since we have no inequality constraints. This proves that (i), (ii) are equivalent.
\end{proof}
\begin{proposition}
\label{foralliffforsome}
    The following are equivalent:
    \begin{enumerate}[label=(\arabic*)]
        \item there exists the optimum  
        for \eqref{imagefunction} for \textbf{all} $\bfy\in\R^M$. 
        \item there exists the optimum  
        for \eqref{imagefunction} for \textbf{some} $\hat\bfy\in\R^M$.
    \end{enumerate}
\end{proposition}
\begin{proof}
    Clearly (1) implies (2). As to the converse, observe that since $f$ is essentially smooth, by Theorem \ref{Ro7026.5} $\nabla f(\R^N)=\mathrm{int}\,\mathrm{dom}f^*\subseteq \mathrm{ri}\,\mathrm{dom}f^*$. Then, by the equivalence established in Proposition \ref{prop:KuhnTucker}, $A^T\lambda(\hat\bfy)=\nabla f(\bfx(\hat\bfy))\in\mathrm{int}\,\mathrm{dom}f^*\subseteq \mathrm{ri}\,\mathrm{dom}f^*$ and Theorem \ref{Ro7016.3}  yields the existence of the optimum 
    for every $\bfy \in \R^M$. 
\end{proof}

\begin{proposition}
    \label{prop:explicit formula}
    Under Assumption 
    \ref{assexput}, the map $\Af:\R^M\rightarrow \R$ is continuously differentiable and strictly convex on $\R^M$. Its conjugate is given by $(\Af)^*(\bfz)=f^*( A^T\bfz),\,\bfz\in\R^M$, which is continuously differentiable on the interior of its domain. The gradient $\nabla \Af(\cdot)$ is a homeomorphism between $\R^M$ and $\mathrm{int}\,\mathrm{dom}(\Af)^*=\mathrm{int}\{\bfz\in\R^M\mid A^T\bfz\in\mathrm{dom}f^*\}=:\mathcal{O}\subseteq \R^M$, and its (continuous) inverse is given by $\nabla (\Af)^*(\bfz),\bfz\in\mathcal{O}$. Finally, the unique optimum $\bfx=\bfx(\bfy)$ of problem \eqref{imagefunction} is given by 
    \begin{equation}
        \label{formula:opt}
        \bfx=\Theta(\bfy)=\nabla f^*\Big(A^T \cdot\nabla (\Af)(\bfy)\Big)
    \end{equation}
where $\Theta:\R^M \to \R^N$ is continuous on $\R^M$.   
\end{proposition}

\begin{proof}
    \ale{By Remark \ref{exist:opt} and  Proposition \ref{foralliffforsome}} 
    there exists an optimum of \eqref{imagefunction} for all $\bfy\in\R^M$.  In particular $\Af(\bfy)\in\R\,\forall\bfy\in\R^M$. As argued in the proof of Proposition \ref{foralliffforsome} there exists $\bflambda\in \R^M$  s.t. $A^T\bflambda\in\mathrm{int}\,\mathrm{dom}(f^*)\subseteq \mathrm{ri}\,\mathrm{dom}(f^*)$. Since $f$ is essentially smooth on the whole $\R^N$, by \cite{Ro70} Corollary 26.3.3 $\Af$ is itself essentially smooth throughout the whole $\R^M$. Existence of optima yields by standard arguments the strict convexity of $\Af$, which is induced by the one of $f$, and $(\Af,\R^M)$ is then of Legendre type. By Theorem \ref{Ro7026.5} applied to $h=\Af$, 
    $\Af:\R^M\rightarrow \R$ is continuously differentiable on $\R^M$. Its conjugate $(\Af)^*$ is continuously differentiable on $\mathcal{O}$, the gradient $\nabla (\Af)(\cdot)$ is a homeomorphism between $\R^M$ and $\mathcal{O}$, and its (continuous) inverse is given by $\nabla (\Af)^*(\cdot )$.
    Now, fix $\bfy\in\R^M$ and take $({\bfx},\bflambda)=(\bfx(\bfy),\bflambda(\bfy))\in \R^N\times \R^M$ solving \eqref{eq:KuhnTucker}. In particular  $A^T\bflambda\in\mathrm{int}\,\mathrm{dom}(f^*)$, and  $\bfx=(\nabla f)^{-1}(A^T\bflambda)=\nabla f^*(A^T\lambda)$, by Theorem \ref{Ro7026.5}. Since $A\bfx=\bfy$, we get $\bfy=A\nabla f^*(A^T\lambda)$. The last step is to prove that $\lambda\in\mathrm{int}\,\mathrm{dom}(\Af)^*=\mathcal{O}$ and $A\nabla f^*(A^T\lambda)=\nabla (\Af)^*(\bflambda)$, as this would give    $\lambda=(\nabla (\Af)^*)^{-1}(\bfy)=\nabla (\Af)(\bfy)$ by Theorem \ref{Ro7026.5} so that $\bfx=\nabla f^*\Big(A^T \lambda\Big)=\nabla f^*\Big(A^T \cdot\nabla (\Af)(\bfy)\Big)$.

    We come to these verifications. First, observe that as argued before $A^T\lambda\in\mathrm{int}\,\mathrm{dom}f^*$, which is open. Then, $\lambda$ belongs to the pre-image of $\mathrm{int}\,\mathrm{dom}f^*$ under $A^T$, which is open by continuity of $A^T$. Set now $\mathcal{O}':=(A^T)^{-1}(\mathrm{int}\,\mathrm{dom}f^*)$. Since $A^T(\mathcal{O}')=\mathrm{int}\,\mathrm{dom}f^*\subseteq \mathrm{dom}f^*$, then $\bflambda\in\mathcal{O}'\subseteq \mathcal{O}$.
    To conclude we prove that $$\nabla (\Af)^*(\bfz)=A\nabla f^* (A^T\bfz),\quad \bfz\in\mathcal{O}'\subseteq \mathcal{O}.$$
    First, by \cite{Ro70} Theorem 16.3 $(\Af)^*(\bfz)=f^*(A^T\bfz)$ for all $\bfz\in\R^M$. The map $(\Af)^*$ on $\mathcal{O}'$ is then the composition of the map $A^T$, differentiable on $\mathcal{O}'$ and taking values in $\mathrm{int}\,\mathrm{dom}f^*$, and $f^*$, differentiable on the latter set. The formula is then the usual chain rule. Continuity of $\Theta$ follows observing that for every $\bfy\in\R^M$, $A^T\bflambda(\bfy)\in\mathrm{int}\,\mathrm{dom}f^*$, where $\nabla f^*$ is continuous, and that $\lambda(\bfy)=\nabla (\Af)(\bfy)$, the latter being continuous on $\R^M$.
\end{proof}

\section{Proof of Theorem \ref{thmfull}}
\label{sec:proffmainthm}

\textbf{ We work again in the Setting \ref{settingbase}}.
By Remark \ref{exist:opt}  we may apply the results in Proposition \ref{prop:explicit formula}.
The proof is indeed based on the following two facts that are proven in Proposition \ref{prop:explicit formula}, using there the notation $f:=-U$, $A\bfX:=\pi(\bfX)$ and $\Af=-\square_\pi U$. \\
    a) The function $\spU : \R^M\rightarrow \R$ is continuous on $\R^M$;\\
    b) Fix any $\bfz \in \R^M$. There exist a unique optimum $\bfx=\Theta(\bfz) \in \R^N$ for the problem $\spU(\bfz)$ with $\Theta :\R^M \to \R^N$ being the continuous function on $\R^M$ defined in \eqref{formula:opt}. In particular we have:   (b1) \, $\spU(\bfz)=U(\Theta(\bfz))$ and (b2) \, $\pi(\Theta(\bfz))=\bfz$.

\textbf{Proof of Item \ref{ThItem1}}.  Fix $\bfZ \in (L^{\infty})^M$. We now prove
\begin{equation}\label{eqAF1}
 \sup\left\{\Ep{U(\bfY)}\mid \bfY\in (L^{\infty})^N, \pi(\bfY)=\bfZ \right\}=:\square_{\pi}\UU(\bfZ)=\Ep{\square_{\pi}U(\bfZ)}. 
\end{equation}
    Fix $\bfz \in \R^M$ and observe that if $\bfy\in\R^N$ satisfies $\pi(\bfy)=\bfz$ then  $$U(\bfy)\leq \sup\left\{U(\bfx)\mid \bfx\in\R^N, \pi(\bfx)=\bfz\right\}=\spU(\bfz)\,.$$
    Now, we can plug in $\bfY\in(\Linfty)^N$ s.t. $\pi(\bfY)=\bfZ$, to get $U(\bfY)\leq\spU(\bfZ)$. From (a) we know that $\spU$, as well as $U$, is a continuous function and so no measurability issues arise and both $U(\bfY)$ and $ \spU(\bfZ)$ are bounded random variables.
    We can then take expectations on both sides of the latter inequality and deduce $(\square_{\pi}\UU)(\bfZ) \leq \Ep{(\spU)(\bfZ)}$. 
We prove the converse inequality. 
    Consider the continuous function $\Theta$ in (b) and set $\hat{\bfY}:=\Theta(\bfZ)$. Then $\hat{\bfY} \in (\Linfty)^N$ and by (b2) above, $\pi(\hat{\bfY})=\pi(\Theta(\bfZ))=\bfZ$.
    Thus, $\hat{\bfY}$ satisfies the constraints in LHS of \eqref{eqAF1}.
 Consequently, 
    \begin{equation*}
     \sup\left\{\Ep{U(\bfY)}\mid \bfY\in (L^{\infty})^N, \pi(\bfY)=\bfZ \right\}\geq \Ep{U\left(\hat{\bfY}\right)}=\Ep{U\Big(\Theta(\bfZ)\Big)}=\Ep{\spU(\bfZ)},   
    \end{equation*}
    by (b1), which concludes the proof of \eqref{eqAF1}.

\textbf{Proof of Item 2}. Recall that in Setting \ref{settingbase} Assumption \ref{mass: F, pi, C} holds true. Fix $\bfZ \in (L^{\infty})^M$ and $\bfX \in (L^{\infty})^N$. From (a) we deduce that $\Ep{\spU(\bfZ)}<\infty$ and by \eqref{eqAF1}, $\square_{\pi}\UU (\bfZ)=\Ep{\square_{\pi} U(\bfZ)}<\infty$.  By the assumption $\sup \{U(\bfy) \mid \bfy \in \mathbb R ^N \}>0$ we obtain $\sup \{\UU(\bfX+\bfy) \mid \bfy \in \mathbb R ^N \}>0$.  Thus \ale{all the assumptions} in Proposition \ref{prop3} are satisfied and hence $\rho_{\pi, \UU}(\bfX)=\rho_{(\square_{\pi}\UU)}(\pi(\bfX))=\rho_{(\Ep{\square_{\pi} U})}(\pi(\bfX))$, by \eqref{eqAF1}. Recalling the definition in \eqref{eqrhog}, we thus proved \eqref{neweq1}. Regarding the optimality of  $\bfY_\bfX:=-\bfX+\Theta(\pi(\bfX)+\widehat{\alpha})$, observe that
$\UU(\bfX+\bfY_{\bfX})=\UU(\Theta(\pi(\bfX)+\widehat{\alpha}))=\Ep{U(\Theta(\pi(\bfX)+\widehat{\alpha}))}=\Ep{\square_{\pi} U(\pi(\bfX)+\widehat{\alpha})}\geq 0$, where in the last equality we used (b1), and the inequality follows from the optimality of $\widehat{\alpha}$ in \eqref{neweq1}. Thus $\bfY_\bfX$ satisfies the inequality constraint in \eqref{eqA}
Moreover, using the linearity of $\pi$ and (b2) we get $$\pi(\bfY_{\bfX})=\pi\Big(-\bfX+\Theta(\pi(\bfX)+\widehat{\alpha})\Big)=-\pi(\bfX)+\pi\Big(\Theta(\pi(\bfX)+\widehat{\alpha})\Big)=-\pi(\bfX)+\pi(\bfX)+\widehat{\alpha}=\widehat{\alpha}$$
so that $\bfY_\bfX \in \mathcal C$. Finally, $\sum_{m=1}^M\pi^m(\bfY_\bfX)=\sum_{m=1}^M\widehat{\alpha}^m=\rho_{\pi, \UU}(\bfX)$, by optimality of $\widehat{\alpha}$. Thus $\bfY_\bfX$ is the desired optimum, which is unique by the strict concavity of $\UU$.

\textbf{Conclusion, for the case $M=1$}. If $\rho_{\pi, \UU}(\bfX)$ is finite, then the optimality of $\widehat{\alpha}=\rho_{\pi, \UU}(\bfX)$ is directly checked by monotone convergence theorem, considering that $\square_{\pi} U$ is continuous on $\R$ and nondecreasing by Proposition \ref{prop:explicit formula}. Thus, we only need to show that $\rho_{\pi, \UU}(\bfX)\in\R$ for every $\bfX\in (\Linfty)^N$. Since we are in Setting \ref{settingbase}, by Remark \ref{remsuponlinfty} we have $\sup \{\UU(\bfX+\bfy) \mid \bfy \in \mathbb R ^N \}>0$ which yields $\rho_{\pi, \UU}(\bfX)<+\infty$. Suppose now by contradiction that $\rho_{\pi, \UU}(\bfX)=-\infty$ and take a minimizing sequence $\bfY_n\in\mcC$ with $\pi(\bfY_n)\downarrow_n-\infty$ and $\Ep{U(\bfX+\bfY_n)}\geq 0$ for every $n$. By Proposition \ref{prop:KuhnTucker}, since we are under Assumption \ref{assexput} and Remark \ref{exist:opt} applies, there exists a $\lambda\in\R$ such that $A^T\lambda\in \nabla f(\R^N)$. By Theorem \ref{Ro7026.5} we have $A^T\lambda\in\mathrm{int dom}(f^*)\subseteq (-\infty,0)^N$, the latter following from monotonicity of $f=-U$, which implies $\lambda<0$: indeed all the components of $A$ are nonnegative, as $A(\R^N_+)=\R_+$. Now by Fenchel inequality we have $-f^*(A^T\lambda)-f(\bfx)\leq (-\lambda)A\bfx$. Substituting $\bfx$ with $\bfX+\bfY_n$ and taking expectations yields a contradiction, as $f^*(A^T\lambda)\in\R$ and $\Ep{-f(\bfX+\bfY_n)}\geq 0$ for each $n$, while RHS tends to $-\infty$, as $\pi(\bfY_n)\downarrow_n-\infty$.

\appendix
\section{Appendix}

A function $\Phi:[0,+\infty)^N\rightarrow \R$ is called multivariate Orlicz function if it
is null in $0$, convex, continuous, increasing in the usual componentwise order and satisfies: there exist $A > 0$, $B$ constants such that 
$\Phi(\bfx)\geq A\sum_{j=1}^Nx^j-B$ for every $\bfx\in [0,+\infty)^N$. We refer to \cite{Drapeau} and \cite{DoldiFrittelli2022} for further details. Inspired by \cite{DoldiFrittelli2022} Definition 3.4, we say that a function $U:\R^N\rightarrow \R$ is well controlled if
there exist a multivariate Orlicz function $\widehat{\Phi}: R^N\rightarrow \R $ and a function $h :[0,+\infty)\rightarrow \R$ such that $ U(\bfx)\leq -\widehat{\Phi}((x)^-)+\varepsilon \sum_{j=1}^N\abs{x^j}+h(\varepsilon)$ for every $\varepsilon>0$.
\begin{lemma}
    \label{lemmaexistsoptu}
    Suppose $U:\R^N\rightarrow \R$ is strictly concave, strictly increasing in the componentwise order and also well controlled. Suppose also that $\pi$ satisfies Assumption \ref{mass: F, pi, C} (a), and that $\sum_{m=1}^M\pi^m(\bfx)=\sum_{j=1}^Nx^j$ for every $\bfx\in\R^N$. Then Assumption \ref{assexput} is satisfied.
\end{lemma}
\begin{proof}
    Take $\bfz=\bfzero\in\R^M$ and take a maximizing sequence $(\bfx_n)_n$ for $\square_{\pi} U(\bfzero)$, w.l.o.g. assuming that $U(\bfx_n)\geq \square_{\pi} U(\bfzero)-1$ for every $n$. By \cite{DoldiFrittelli2022} Lemma 3.5.(iv) we have for some $a>0,b\in\R$ that
    \begin{align*}
     &\square_{\pi} U(\bfzero)-1\leq   U(\bfx_n)\leq  a\sum_{j=1}^N(x^j_n)^+-2a\sum_{j=1}^N(x^j_n)^-+b\\
     &=a\sum_{j=1}^Nx^j_n-a\sum_{j=1}^N(x^j_n)^-+b=a\sum_{m=1}^M\pi^m(\bfx_n)-a\sum_{j=1}^N(x^j_n)^-+b=a\sum_{m=1}^Mz^m+b-a\sum_{j=1}^N(x^j_n)^-
    \end{align*}
It follows that $\sum_{j=1}^N(x^j_n)^-$ needs to be bounded, and since also $\sum_{j=1}^N(x^j_n)^+=\sum_{j=1}^Nx^j_n+\sum_{j=1}^N(x^j_n)^-=\sum_{m=1}^Mz^m+\sum_{j=1}^N(x^j_n)^-$ the same holds for $\sum_{j=1}^N(x^j_n)^+$. Thus $(\bfx_n)_n$ is bounded in $\R^N$. Passing to a subsequence converging to some $\bfx_\infty\in\R^N$ we get by continuity of $\pi$ (which is linear on $\R^N$ and takes values in $\R^M$ by hypothesis) that $\pi(\bfx_\infty)=\bfz$, and since $U$ is continuous on $\R^N$ (by \cite{Ro70} Theorem 10.4 applied to $f=-U$, since it is finite-valued on the whole $\R^N$ by assumption) we have $U(\bfx_\infty)=\lim_nU(\bfx_n)=\square_{\pi} U(\bfzero)$. This proves the optimality of $\bfx_\infty$.
\end{proof}

\bibliographystyle{abbrv}
\bibliography{BibAll}

\end{document}